	\documentclass[conference]{IEEEtran}

\usepackage{cite}
\ifCLASSINFOpdf

\else

\fi
\usepackage{subfig}
\usepackage{bm}
\usepackage{amsmath}
\usepackage{epsf}
\usepackage{graphics}
\usepackage{ amssymb }
\usepackage[dvips]{graphicx}
\usepackage{epsfig}
\usepackage{cite}
\usepackage[linesnumbered,ruled,vlined]{algorithm2e}
\usepackage{colortbl}
\usepackage{color}
\usepackage{multicol}
\usepackage{bm}
\usepackage{amsmath}
\usepackage{epsf}
\usepackage{graphics}
\usepackage{ amssymb }
\usepackage[dvips]{graphicx}
\usepackage{epsfig}
\usepackage{cite}
\usepackage[linesnumbered,ruled,vlined]{algorithm2e}
\usepackage{graphicx}
\usepackage{epsfig}
\usepackage{latexsym}
\usepackage{amsfonts}
\usepackage{here}
\usepackage{rawfonts}

\usepackage[utf8]{inputenc}
\usepackage[english]{babel}
\usepackage{amsmath}
\usepackage{amsfonts}
\usepackage{amssymb}
\usepackage{color} 
\usepackage{bm}
\usepackage{listings}
\usepackage{caption}
\usepackage{amssymb}
\usepackage{amsthm}
\usepackage{graphicx}
\usepackage{epstopdf}
\usepackage{listings}
\usepackage{float}
\usepackage{amsmath}
\usepackage{amssymb}
\usepackage{amsfonts}
\usepackage{epstopdf}

\usepackage{multirow}
\usepackage{amscd}
\usepackage{mathrsfs}
\usepackage{graphicx}
\usepackage{color}
\usepackage{url}
\usepackage{bm}
\usepackage{setspace}

\usepackage{footnote}
\newtheorem{theorem}{Theorem}

\newtheorem{definition}{Definition}

\usepackage[]{algorithm2e}
\renewcommand{\paragraph}{\@startsection{paragraph}{4}{\z@}%
	{-2.00ex\@plus -1ex \@minus -.2ex}%
	{0.5ex \@plus .2ex}%
	{\normalfont\normalsize}}
	
 \usepackage[protrusion=true,expansion=true]{microtype}

\IEEEoverridecommandlockouts
\usepackage{cite}
\usepackage{amsmath,amssymb,amsfonts}
\usepackage{algorithmic}
\usepackage{graphicx}
\usepackage{textcomp}
\usepackage{xcolor}
\def\BibTeX{{\rm B\kern-.05em{\sc i\kern-.025em b}\kern-.08em
    T\kern-.1667em\lower.7ex\hbox{E}\kern-.125emX}}
    
\usepackage[inner=0.675in,outer=0.675in,top=0.7in,bottom=1in,pdftex]{geometry}\newcommand\semiHuge{\@setfontsize\semiHuge{22.72}{27.38}}
\renewcommand{\paragraph}{\@startsection{paragraph}{4}{\z@}%
	{-2.00ex\@plus -1ex \@minus -.2ex}%
	{0.5ex \@plus .2ex}%
	{\normalfont\normalsize}}
\begin{document}

\title{{Reinforcement
	Learning for Interference Avoidance Game in RF-Powered  Backscatter Communications}\thanks{This work was supported in part by the National Science Foundation under grants ECCS-1444009 and CNS-1824518.}  }
	
\author{\IEEEauthorblockN{Ali Rahmati, and Huaiyu Dai}
\IEEEauthorblockA{{$^*$Department of Electrical and Computer Engineering}, 
{NC State University,}
Raleigh, NC, US \\
Email:\texttt{ \{arahmat@ncsu.edu, hdai@ncsu.edu}\}}
}


\maketitle

\begin{abstract}
RF-powered backscatter communication is a promising new technology that can be deployed for battery-free applications such as internet of things (IoT) and wireless sensor networks (WSN). However, since this kind of communication is based on the ambient RF signals and battery-free devices, they are vulnerable to interference and jamming. In this paper, we model the interaction between the user and a smart interferer in an ambient backscatter communication network as a game. We design the utility functions of both the user and interferer in which the backscattering time is taken into the account. The convexity of both sub-game optimization problems is proved and the closed-form expression for the equilibrium of the Stackelberg game is obtained.  Due to lack of information about the system $\textrm{SNR}$ and transmission strategy of the interferer, the optimal strategy is obtained using the  Q-learning algorithm in a dynamic iterative manner. We further introduce hotbooting Q-learning as an effective approach to expedite the convergence of the traditional Q-learning. Simulation results show that our approach can obtain considerable performance improvement in comparison to random and fixed backscattering time transmission strategies and improves the convergence speed of Q-Learning by about~$31\%$.
\end{abstract}

\IEEEpeerreviewmaketitle

\section{Introduction}
\noindent Ambient backscatter communication has been introduced recently, which can provide the possibility of communication between two nodes using available ambient RF signals without using any active radio transmission \cite{lu2015wireless,van2018ambient}. Thus, it can be a promising technology for beyond 5G \cite{parvez2018survey} which is low-cost and energy efficient  for battery-free applications such as internet of things (IoT), drones~\cite{rahmati2019dynamic}, MIMO communication \cite{yang2016physical, amiri2018extremely} and sensor networks. In such networks, the user tags can harvest energy from the ambient signals and transmit the information to the receiver over ambient RF carriers without using any  dedicated signals.

However, one of the drawbacks of such systems is that they are sensitive to the interference and jamming due to poor quality of the backscattered signal. Thus, in the presence of the smart interferer, the user should try to choose the best strategy for the transmission. However, in practical scenarios, the user does not have any information about the state of the system and smart interferer strategy. As a result, a dynamic scenario should be considered in which the user can learn its transmission strategy via trial-and-error in a dynamic manner.
This paper is prompted to tackle  such an issue.

In the literature, there are several papers that considered the time allocation for the backscattering in hybrid wireless powered systems. For instance, in \cite{hoang2017overlay}, a game theoretic formulation for RF-powered backscatter cognitive radio networks is presented in which the secondary user chooses its backscattering time and the gateway charges the user accordingly. In \cite{hoang2017optimal}, the optimal time sharing in a RF-powered backscatter cognitive radio network is investigated in order to maximize the overall network throughput. In \cite{thai2016tradeoff}, the integration of the ambient backscatter and a cognitive radio network is introduced. The trade off between the harvest-then-transmit and backscatter communication is investigated with the goal of maximizing the overall transmission rate.
Thus, backscattering time plays a critical role in such systems. 

On the other hand, game theory and Q-learning are used for interference avoidance scenarios in the literature to design the optimal power allocation for the players. In order to design the utility function in ambient backscatter systems, the backscattering and energy harvesting time should be considered as they affect the number of transmission bits considerably. As a result, assuming the signal-to-interference-plus-noise ratio (SINR) as the utility of the players can not describe the scenario well enough. In this paper, we define the utility of the user and the smart interferer as the number of bits and prove that both the user and the smart interferer sub-games are convex optimization problems, and the best strategy of each is obtained. 
Moreover, a scenario including the  dynamic interaction  between the user and the smart interferer based on ambient backscatter communication is  considered, which is not studied in the literature. In \cite{8219367}, the authors    formulated  such an interaction  as a Stackelberg game in which  the players are aware of the system parameters. However, in practice, the optimal performance depends on the system state and jamming parameters which are challenging for the user to accurately estimate.  Thus, a more practical scheme is needed in which the system can learn the best strategy throughout the time.

The main contributions of the paper are as follows:

\begin{itemize}
\item An interference avoidance game for the ambient backscatter communication is modeled, in which the user selects its portion of time for backscattering and the smart interferer chooses its transmission power.

\item By taking the backscattering time into account, the utility of the both the user and the smart interferer is defined as the number of transmitted bits. The convexity of the optimization problems for each sub-game is proved and the best strategy for each of them is obtained.
\item Due to uncertainty of the system state and smart interferer transmission policy, both user and the smart interferer 
can deploy Q-learning algorithm   in order to obtain the optimal policy for the dynamic game.

\item We adopted the hotbooting Q-learning to further improve the learning speed of the algorithm. 

\end{itemize}
The rest of the paper is organized as follows. In Section \ref{sec2}, the system model and game formulation are presented.
In Section \ref{sec3}, in order to achieve the optimal policy of the proposed dynamic game,  the Q-learning algorithm is discussed. Hotbooting Q-learning is applied to the scenario to accelerate the learning process. Section \ref{sim} investigates the performance of our scheme via simulation results, and finally Section \ref{con} concludes the paper.
\section{System Model and Problem Formulation}\label{sec2}
\subsection{System Model}
\noindent We consider an RF-powered backscatter network in which there is a user in the system  equipped with an RF energy harvesting module and a backscatter circuit to harvest energy from the available RF signal transmitted by the  hybrid access point (HAP) or backscatter the transmitted signal to the HAP. There is a smart interferer that can learn the  transmission strategy of the  user and respond accordingly during a sequence of time slots. The smart interferer aims at decreasing the number of backscattered bits by the user. It is assumed that the time is divided into equal slots as $n \in \{1, 2, \dots, T\}$. The transmission happens in a slotted manner and we assume that channel parameters are fixed during each slot $n$. Let us assume the transmission in a normalized time slot with the duration of 1.  The time portion allocated to backscatter communication at slot $n$ is denoted by $\phi_n$ and the rest of the time, i.e., $1-\phi_n$ is allocated for energy harvesting. The received RF power from the HAP to the user in the free space can be written as follows:
\begin{equation}
P^\textrm{R}=\delta P^{\textrm{HAP}}\frac{G^TG^R \lambda^2_\textrm{HAP}}{(4\pi d_\textrm{HAP})^2},
\end{equation}
where  $P^{\textrm{HAP}}$ is the transmit power of the HAP, $\delta$ is the energy harvesting efficiency, $G^T$ is the antenna gain of the HAP and $G^R$ is the antenna gain of
the user, $\lambda_\textrm{HAP}$ is the wavelength of the signal transmitted from the HAP, and
$d_\textrm{HAP}$ is the distance between the HAP and the user tag  \cite{talla2013hybrid}.
Thus, we can derive the total amount of harvested energy at the user in time slot $n$ as:
\begin{equation}
E^\textrm{h}_n=(1-\phi_n) P^{\textrm{R}}.
\end{equation}
After energy harvesting in the first phase, the user exploits all of the energy to send its own data to the receiver during the backscatter time, i.e., $\phi_n$. Therefore, the average transmit power of the user at slot $n$ is given by:
\begin{equation}
P^\textrm{T}_n=\frac{E^\textrm{h}_n}{\phi_n}.
\end{equation}
\begin{figure}[!t]
\vspace{3mm}
	\includegraphics[width=7cm,height=5cm]{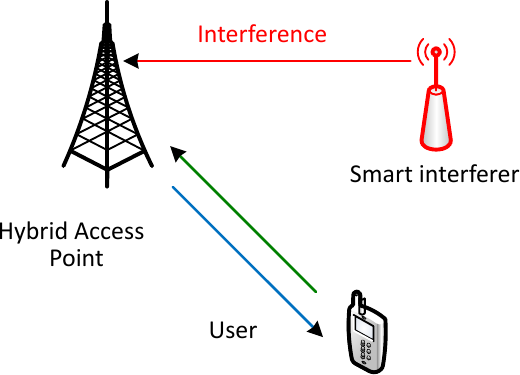}
	\centering
	\caption{System model.}
	\label{fig:1}
\end{figure}
\noindent As a result the received $\textrm{SINR}_n$ at the HAP  at time slot $n$ can be written as:
\begin{equation}
\textrm{SINR}_n=\frac{P^\textrm{T}_n h |\Gamma_0-\Gamma_1|^2}{p_n^Jg+N_0},
\end{equation}
where  $\Gamma_1$ and  $\Gamma_0$ are the reflection coefficients, $N_0$ is the noise power, $p_n^{J}$ is the jamming  power, while $h={G^TG^R \lambda^2_\textrm{HAP}}/{(4\pi d_\textrm{HAP})^2}$ and $g={G^JG^T \lambda^2_J}/{(4\pi d_J)^2}$ are the channel gain from the user to HAP and from the smart interferer to HAP, respectively.
Here, $G^J$ is the antenna gain of the  smart interferer, $\lambda_J$ is the
wavelength of the signal transmitted from the smart interferer, and
$d_J$ is the distance from the smart interferer to the  user. In backscatter communications systems, a
two-state modulation is typically used because of its simplicity. This can be implemented using load modulation \cite{van2018ambient}.
By changing the antenna impedance, the reflection coefficient can be shifted between the
absorbing and reflecting states, respectively. If the impedance is set for the absorbing state, the impedance is  matched and the RF signal can be absorbed. In contrast, if the impedance is not matched, the RF signal is backscattered and can be assumed as bit `1' in the receiver.
Having the $\textrm{SINR}_n$, the backscattering transmit data rate in time slot $n$ can be obtained as:
\begin{equation}
r^B_n=\kappa W \log_2 (1+\textrm{SINR}_n),
\end{equation}
where $\kappa \in [0,1]$ is the transmission efficiency, and $W$ is the channel bandwidth.
Based on this, the number of transmitted bits during the backscattering time $\phi_n$ in time slot $n$ is given by:
\begin{align}\label{qq111}
R^B_n &=  \phi_n\kappa W \log_2 \left(1+\textrm{SINR}_n\right)\\ \nonumber
&=\phi_n\kappa W \log_2 \left(1+\frac{(1-\phi_n) h |\Gamma_0-\Gamma_1|^2}{\phi_n(p_n^Jg+N_0)}\right).
\end{align}
As can be seen in \eqref{qq111}, considering the number of transmitted bits as a part of the utility function can make the problem more challenging due to the multiplication of the $\phi_n$ to the transmit data rate.

The interactions between the smart interferer and the user can be formulated as an  interference avoidance game for ambient backscatter communications. In the next section, the game formulation is presented. 
\subsection{Game Formulation}

The dynamic interference avoidance Stackelberg game between the smart interferer and the user is formulated as a game denoted by $\mathbb{G}$, in which  the user as the leader chooses   the portion of time for backscattering in order to improve its utility, and  the smart interferer as a follower chooses its jamming power to reduce the utility of the user as much as possible. The utility of the user denoted by $U_U$  depends on the ambient backscattered number of transmitted bits and the cost of backscattering offered by HAP as a gateway \cite{hoang2017overlay},  which is given by
\begin{equation}\label{eqw}
U_{{U}}(\phi_n, {p^J_n})=\phi_n\kappa W \log_2 \left(1+\textrm{SINR}_n\right)-C_\phi \phi_n,
\end{equation}
where $C_\phi$ is the price charged by the HAP for the backscattering time. The utility of the smart interferer denoted by $U_J$ can be defined as
\begin{equation}\label{eqq3}
U_{{J}}(\phi_n, {p^J_n})=-\phi_n\kappa W \log_2 \left(1+\textrm{SINR}_n\right)-C_J p_n^J,
\end{equation}
where $C_J$ is the smart interferer unit price for power transmission.
In the game, the user   chooses its strategy in order to maximize the utility given in \eqref{eqw}, while  the smart interferer aims to maximize the  utility given in \eqref{eqq3}, which is equivalent to minimizing the backscattered number of transmission bits by the user subject to the jamming cost. In summary, in time slot $n$, the game can be formulated as $\mathbb{G}=\langle\{U,J\}, \{\phi_n,p^J_n\},\{U_U,U_J\}\}\rangle$, in which the user chooses its backscattering time portion $\phi_n$ in each time slot, while the smart interferer chooses its transmission power. The Stackelberg game  can be formulated as:
\begin{equation}
p^{J*}_n({\phi_n})= \arg ~\max_{0\le p_n^J \le p^{J,\max}_n} ~ U_J(\phi_n, p^J_n),
\end{equation}

\begin{equation}
{\phi}^*_n= \arg ~\max_{0\le \phi_n \le 1} ~ U_U(\phi_n, p^{J*}_n({\phi_n})).
\end{equation}In the following, the definition of the Stackelberg equilibrium is given. Then, the equilibrium of the proposed game is obtained.
\subsection{Stackelberg Equilibrium (SE)}
In what follows, the SE of the game is obtained through backward induction.
\begin{definition}

The SE of the proposed game $\mathcal{G}$ in time slot $n$ is denoted by the policy $(\phi^{*},p^{J*}_n)$, for which  the following conditions hold
\begin{equation}
U_U(\phi^*_n, {p^{J*}_n}) \ge U_U(\phi_n, {p^{J*}_n}),~ \forall \phi_n \in [0,1],
\end{equation}
\begin{equation}
U_J(\phi^*_n, {p^{J*}_n}) \ge U_J(\phi^*_n, {p^{J}_n}),~ \forall {p_n^{J}}, \in [0,p_n^{J,\textrm{max}}].
\end{equation}
\end{definition}
In order to find the SE of the game, we can find its sub-game perfect Nash Equilibrium (NE). NE
is defined as the  point(s) at which no player can
improve its utility by changing its strategy unilaterally, while
other players' strategy is not changed \cite{rahmati2015price},~\cite{rahmati2016price}. To do this, the well-known backward induction method is employed \cite{fudenberg1991game,rahmati2017price}. In this method, at first, given the leader's strategy, the best response of the follower is obtained by solving the follower's sub-game. Then, substituting the best response of the follower, the best response of the leader is obtained.
\begin{theorem}
Given the backscattering time $\phi_n$, there exists a unique optimal solution for the smart interferer's sub-game, and  the best response strategy  for the SE  game at the smart interferer side is given by:
\begin{equation}\label{ww}
p^{J*}_n=\dfrac{\sqrt{\ln\left(2\right)}\sqrt{\Delta_n}+\ln\left(2\right)C_JD_n}{2\ln\left(2\right)C_J \phi_n g}
\end{equation}
where $D_n=(\phi_n-C_J)h |\Gamma_0-\Gamma_1|^2-2\phi_n N_0$ and    $\Delta_n=\ln\left(2\right)C_J^2(\phi_n-1)^2h^2|\Gamma_0-\Gamma_1|^2+4C_J\phi_n^2\left(1-\phi_n\right)gh\kappa W|\Gamma_0-\Gamma_1|^2$.

\end{theorem}
\begin{proof}
In order to investigate the concavity of the smart interferer's objective function, the second derivative of the smart interferer's utility function is obtained as:
\begin{align}\nonumber
\frac{\partial^2 U_J(\phi_n,p^J_n)}{\partial (p_n^J)^2  }&=\\&\hspace{-14mm}-\frac{(1-\phi_n)\phi h\kappa W |\Gamma_0-\Gamma_1|^2(2\phi_n^Jg p_n^J+B_n)}{\ln(2)(p_n^J g+N_0)^2(\phi_n^Jg+L_n)},
\end{align}
in which $V_n=(1-\phi_n^J)h|\Gamma_0-\Gamma_1|^2+\phi_nN_0$ and $L_n=(1-\phi_n^J)h|\Gamma_0-\Gamma_1|^2+2\phi_nN_0$. Considering $0 \le \phi_n \le 1 $, the second derivative of the smart interferer's utility function is negative, i.e., $\frac{\partial^2 U_J(\phi_n,p^J_n)}{\partial (p_n^J)^2  }\le 0$, which leads to concavity of the smart interferer sub-game. As a result, given the $\phi_n$, the optimal solution for the smart interferer's sub-game is available. By setting the first derivative of \eqref{eqq3} equal to zero, one can get the following quadratic equation:
\begin{equation}\label{quad}
A_n(p_n^J)^2+B_np_n^J+C_n=0,
\end{equation}
where $A_n=\ln(2) \phi_n C_J g^2, B_n=2\ln(2) \phi_n C_J N_0 g+\ln(2)(1- \phi_n) g C_J h|\Gamma_0-\Gamma_1|^2, C_n=\ln(2) \phi_n C_J N_0^2+\ln(2) (1-\phi_n) C_Jh|\Gamma_0-\Gamma_1|^2- \phi_n(1-\phi_n)gh\kappa W |\Gamma_0-\Gamma_1|^2$. By solving the quadratic equation in \eqref{quad} and choosing the positive root,
one can get the optimal smart interferer's transmission power as in \eqref{ww}. It can be observed that for the given range of $0 \le \phi_n \le 1 $, we can get $\Delta_n \ge 0$ which leads to the real solution.
\end{proof}
Assuming the user as the leader of the game, it can obtain its best strategy given the follower's action using the following Theorem.
\begin{theorem}
Given the smart interferer's action, there exists a unique optimal solution for the user's sub-game, and  the best response strategy  for the SE  game at the user side can be obtained in a numerical manner.
\end{theorem}
\begin{proof}
In order to investigate the concavity of the user's sub-game, the second derivative of the user's utility function is obtained as:
\begin{equation}\label{ssaa}
\frac{\partial^2 U_U(\phi_n,p^J_n)}{\partial \phi_n^2  }=-\dfrac{h^2\kappa W |\Gamma_0-\Gamma_1|^2}{\ln\left(2\right)\phi_n H_n^2},
\end{equation}
where $H_n=\left(h |\Gamma_0-\Gamma_1|^2 -gp^J_n-n\right)\phi-h|\Gamma_0-\Gamma_1|^2$. Based on \eqref{ssaa}, we can conclude that ${\partial^2 U_U(\phi_n,p^J_n)}/{\partial \phi_n^2  }\le 0$, which indicates that the user's utility function is concave with respect to $\phi_n$.  Hence, given  $p_n^J$, the optimal solution for the user's sub-game is available. By setting the first derivative of \eqref{eqw} equal to zero,  the optimal portion of backscattering time $\phi_n^*$ can be achieved by solving the following equation:
\begin{align}\nonumber
 \frac{\partial U_U(\phi_n,p^J_n)}{\partial \phi_n  }\bigg|_{\phi_n=\phi^*_n}=\dfrac{\kappa W\ln\left(\frac{h|\Gamma_0-\Gamma_1|^2\left(1-\phi_n^*\right)}{\left(g p^J_n +N_0\right)\phi_n^*}+1\right)}{\ln\left(2\right)}\\ \nonumber
\hspace{-0.2cm}+\dfrac{\kappa W\left(-h|\Gamma_0-\Gamma_1|^2\phi_n^*-h|\Gamma_0-\Gamma_1|^2\left(1-\phi_n^*\right)\right)}{\ln\left(2\right)\left(\left(g p^J_n+N_0\right)\phi_n^*+h|\Gamma_0-\Gamma_1|^2\left(1-\phi_n^*\right)\right)}-C_\phi =0.
\end{align}
However, due to non-linearity of the obtained equation, the closed form expression is not available, and the optimal solution can be obtained numerically using well-known methods such as Newton's method~\cite{atkinson2008introduction}.
\end{proof}

The best responses of the both sub-games with respect to smart interferer's transmission power and the portion of backscattering time  depends on
the distance of the the smart interferer to the HAP, the  channel states and the parameters of the smart interferer in the time slot, which are not straightforward for the user to accurately estimate. 
In practice, since the smart interferer and the user are not aware of the interference power and channel gain, computation of the best responses is challenging. One approach to deal with this problem is to model their interaction as  a repeated dynamic anti-jamming  transmission
game and resort to learning algorithms to approach the optimal policies \cite{hamidouche2018collaborative, kasgari2018dynamic}. 

\section{Dynamic Interference Avoidance Game for Backscatter Communication}\label{sec3}
\noindent In dynamic  environments,  optimal strategy selection depends on the radio channel state and jamming parameters, which are challenging for the user to accurately estimate. Thus, a more practical scheme is needed in which the system can learn the best strategy through the time.  Here, the repeated interaction between the user and the smart interferer can be considered as a dynamic game, in which the user chooses its portion of time for backscattering according to the system state and history of the smart interferer transmission, and the smart interferer acts based on the previous backscatter transmission information. The user interaction in  the proposed game can be considered as a learning-based game and the user can deploy reinforcement learning to obtain the best strategy \cite{xiao2018two, kasgari2018dynamic}.

The feasible set for the backscattering time of the user  is quantized into $K+1$ levels, i.e., $\phi_n^k \in [\phi^{\textrm{max}}k/K]_{0\le k\le K} = \Phi$ where $\phi^{\textrm{max}}$ is the maximum portion of time for backscattering. Similarly, for the transmit power of the smart interferer, we have  $p_n^J \in [p^{\text{max}}m/M]_{0\le m\le M} $,  where $p^{\text{max}}_J$ is the   maximum transmit power of the smart interferer.  In addition, the user quantizes the
received SINR into one of $L$ levels as well, and sends it
back  to the HAP.
It is assumed that the user and the smart interferer
are intelligent agents. 
 We consider that both the user and the smart interferer can adopt a mixed strategy, which can be interpreted as a probability distribution for the strategy set of each player. Let $\bm{\pi}_u$ and $\bm{\pi}_J$ denote the mixed strategy of the user and the smart interferer, respectively. Applying Q-Learning, such distributions can be obtained.

 \subsection{Q-Learning}
 The  Q-learning algorithm can be used for solving the proposed game in a repeated manner over the time. In order to solve the game, we can deploy Q-learning without knowing the state of the system, which is the system SINR. The state of the system at time $n$  can be represented as $s^n={[\textrm{SINR}^{n-1}_l]}_{1\le l\le L}  \in ~ \mathcal{L}$, where $\mathcal{L}$ is the set of all possible SINR values.

 \begin{algorithm}[h]

	1: $\alpha, \beta, \gamma$, $Q(s, p) = 0, V(s) = 0, \forall s, \forall~p$, $ s_0=0$ , $\pi=\frac{1}{|\Phi|}$ 
	
	 \For{ $i = 1, 2, ..., I$ }{  Generate a realization of the similar environment.
     
        \For{ $n = 1, 2, ..., N$ }{
     Select $\phi_n \in \Phi$ randomly.
     
     Observe the utility $U_u$ and the $\textrm{SINR}_n$ of the system.
     
     $s_{n+1}=\textrm{SINR}_n$
     
     $ Q^*(s_n,\phi_n)\leftarrow (1-\beta)Q^*(s_n,\phi_n)+\beta(U_u(s_n,\phi_n)$ 
     
    $~~~~~~~~~~~+\gamma V^*{(s_{n+1}})).
$
		 
	$	V^*(s_n)=\max_{\phi \in \Phi} ~ Q^*(s_n, \phi_n)$

     } }

	\caption{{Hotbooting process}}
	\label{alg}
\end{algorithm}
 The learning rate denoted by $\beta \in (0,1]$ represents  the weight corresponding to the current executed action, and the discounting factor $\gamma \in [0,1]$ is used to encourage the user to prefer maximizing the closer utility rather than the future ones. The Q-function denoted by $Q(s,\phi_n)$ for state $s$ and action $\phi_n$ can be updated in each time slot $n$ based on the iterative Bellman equation given by
  \begin{equation}\label{qq22}
 Q(s_n,\phi_n)\leftarrow (1-\beta)Q(s_n,\phi_n)+\beta(U_u(s_n,\phi_n)+\gamma V{(s_{n+1}})).
 \end{equation}
 At each time slot $n$, the value function denoted by $V(s_n)$ can be obtained by maximizing the Q-function $Q(s,\phi_n)$ over the set of user actions, i.e., $\Phi$, which is given by
  \begin{equation}\label{qq11}
V(s_n)=\max_{\phi \in \Phi} ~ Q(s_n, \phi_n).
 \end{equation}
 Let us denote the mixed-strategy probability  with $\pi(s,\phi_n)$  which can be updated using the $\epsilon$-greedy algorithm in order to adjust the tradeoff between exploration and exploitation. The updating rule for this  probability  can be written as
 \begin{equation}\label{qqee}
 \pi(\phi_n=\tilde{\phi}_n)=\begin{cases}
1-\epsilon, & \text{ if } \tilde{\phi}_n= \textrm{argmax}_{\phi_n \in |\Phi|} ~Q(s_n,\phi_n)) \\ 
\frac{\epsilon}{|\Phi|-1}, & \text{ if } \textrm{otherwise} 
\end{cases}.
 \end{equation}
 Without knowing  the instantaneous state of the system, the user can learn the best transmission strategy based on the transmission history of the smart interferer which enables the user to improve its long-term efficiency. The Q-learning algorithm suffers from the slow convergence. However, the convergence speed can be improved using the so-called hotbooting Q-learning.

 \subsection{Hotbooting Q-Learning}
 Since the initial values of the Q-tables are all zeros, in order to get the optimal policy, it wastes the time due to random exploration \cite{kaelbling1996reinforcement,xiao2018reinforcement}. 
  Thus,  a hotbooting approach is proposed in order to expedite the convergence. In this approach, the Q-values are initialized using a training dataset which is already collected from the realization of the same environment.  In particular, the Q-values and the mixed strategy probabilities are  initialized using $Q^*$  and $\pi^*$  obtained in advance from the realizations of the similar  
 scenarios according to Algorithm~1. As a result, the hotbooting Q-learning reduces the random exploration in the beginning iterations and can speed up the learning process and convergence.  The overall Hotbooting Q-learning algorithm is presented in Algorithm 2.

  \begin{algorithm}[h]
	
	\DontPrintSemicolon
	$\alpha, \beta, \gamma$, $Q(s, p) = Q^*(s, p), V(s) = V^*{(s)}, \forall s, \forall~p$, $\pi=\pi^*$.

	 \For{ $n = 1, 2, ...$ }{       
        {
     Select $\phi_n \in \Phi$ based on \eqref{qqee}.
     
     Observe the utility $U_U$ and the $\textrm{SINR}_n$ of the system.
     
     $s_{n+1}=\textrm{SINR}_n$.
     
     Update  $ Q(s_n,\phi_n)$ using \eqref{qq22}.
		 
	Update  $	V^*(s_n)$ using \eqref{qq11}.

     } }

	\caption{{Hotbooting Q-Learning Algorithm}}
	\label{alg}
\end{algorithm}
\section{Simulation Results}\label{sim}

\noindent In this section, the performance of the proposed scheme is evaluated via simulations. In the simulations, we assume that both the user and the smart interferer are stationary and can observe the communication. The high power HAP transmits the signal at $2.1$ GHz with $1$ MHz  bandwidth \cite{7981380}. The transmit power of the un-modulated signal is assumed as $43$dBm. The backscattering time cost and smart interferer transmit power cost are assumed as   $C_\phi=C_J=0.1$. The energy harvesting efficiency, and reflection coefficients are assumed as $\delta=0.5$, $\Gamma_0=1$, and $\Gamma_1=-1$, respectively. The maximum transmit power of the smart interferer is assumed as $p_J^{\textrm{max}}=30$ dBm. It is further assumed $G_t=G_i=6$ dBi,
$G_r=1.8$ dBi \cite{assimonis2016sensitive}. The number of the actions for both the user and the smart interferer is assumed as $K=M=10$. For the benchmark, we consider random allocation and fixed strategy selection for the user. For the random case, the user selects its strategy randomly  at each time slot, while for the fixed strategy it considers equal time portion for backscattering and energy harvesting. The distance between the user and HAP is assumed as $d_\textrm{HAP}=15$m \cite{7981380} while that between   the user and the smart interferer is assumed as $d_J=20$m. In the $\epsilon$-greedy algorithm, we have assumed  $\epsilon=0.05$.

	\begin{figure}[!t]
	\includegraphics[width=8.2cm,height=6.7cm]{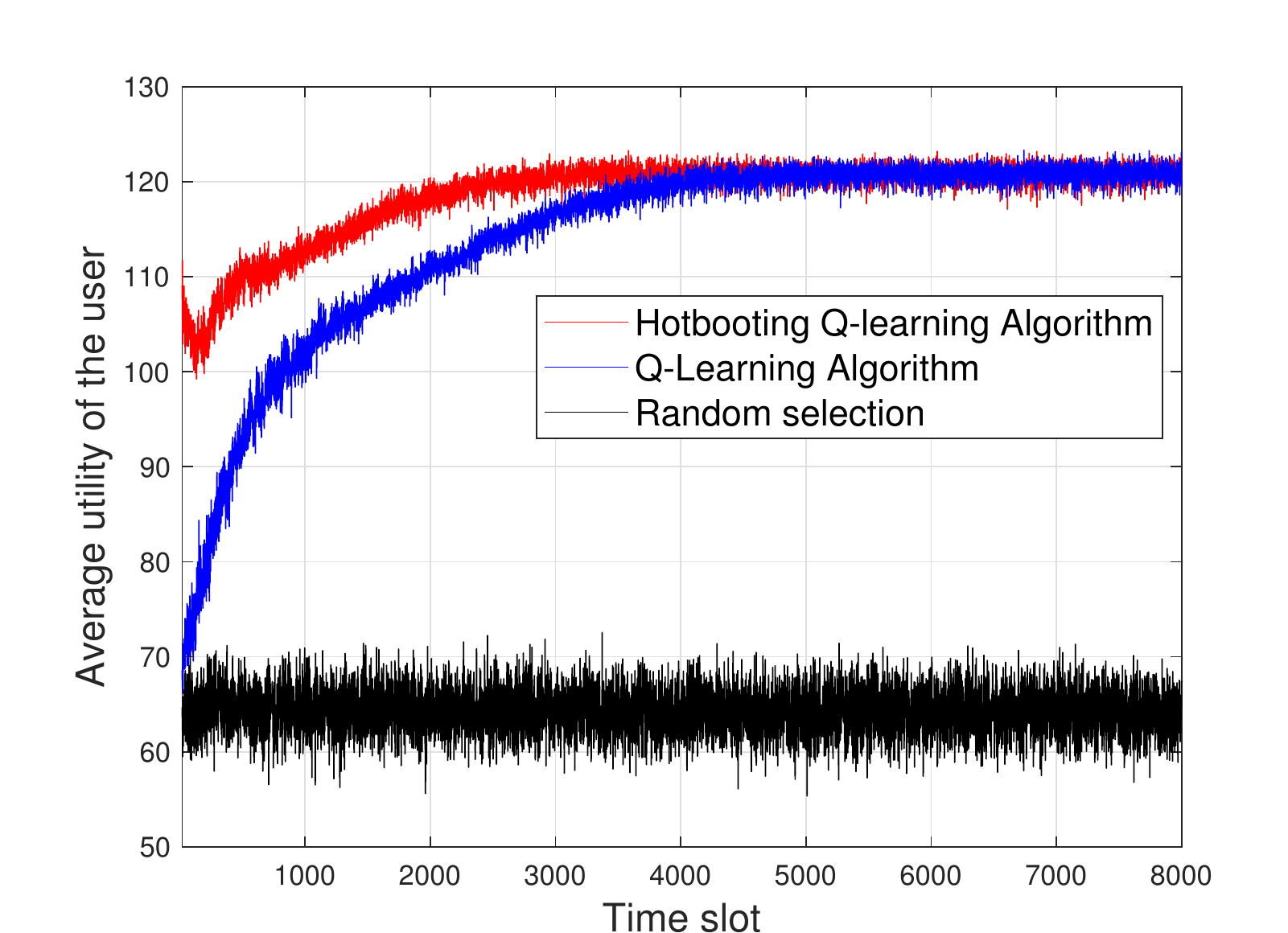}
	\centering
	\caption{Utility of the user vs. time slot to show the convergence performance of the proposed schemes in the dynamic game.}
	\label{conv}
\end{figure}

\begin {table}[h]
\caption {Simulation parameters.} \label{tab1} 
\vspace{-3mm}
\begin{center}
\begin{tabular}{ |c|c|c|c| }
\hline
\textbf{Parameter} & \textbf{Description} & \textbf{Value} \\
\hline
\footnotesize{$f_c$} & \footnotesize{Carrier frequency of HAP} & \footnotesize{2.4 GHz} \\ \hline
\footnotesize{$W$} & \footnotesize{Bandwidth} & \footnotesize{1 MHz} \\ \hline
\footnotesize{$C_\phi$} & \footnotesize{Backscattering time cost } & \footnotesize{0.1} \\ \hline \footnotesize{$C_J$} & \footnotesize{Interferer  power cost} & \footnotesize{0.1} \\ \hline \footnotesize{$\delta$} & \footnotesize{Energy harvesting efficiency} & \footnotesize{0.5} \\ \hline \footnotesize{$\Gamma_0, \Gamma_1$}
& \footnotesize{Reflection coefficients} & \footnotesize{1, -1} \\ \hline \footnotesize{$p_J^{\textrm{max}}$}
 & \footnotesize{Max. interferer   power } & \footnotesize{30 dBm} \\ \hline \footnotesize{$G^T$, $G^R$, $G^J$} & \footnotesize{Antenna gains} & \footnotesize{6dBi, 6dBi, 1.8dBi} \\ \hline 

\end{tabular}
\end{center}
\end {table}
As shown in Fig. \ref{conv}, the utility of the user increases with time and the Q-learning based dynamic solution can considerably improve the performance in comparison to the random scheme. Moreover, the convergence speed of the algorithm can be  further improved using the hotbooting Q-learning. As shown in Fig. 2, the convergence speed of the hotbooting Q-learning is  about $31\%$  better than that    of the Q-learning algorithm. 
Fig. \ref{dist} shows the average utility of the user at varying distances between the HAP and the user. Q-Learning and hotbooting Q-Learning perform almost the same in terms of the average convergence value and both outperform the random and fixed schemes. The utility decreases dramatically when the distance between the user and HAP increases. 
The utility of the user versus the cost of backscattering time is shown in Fig.~\ref{cost}. As can be seen, by increasing the backscattering time cost, the utility of the user decreases.

\begin{figure}[!t]
	\includegraphics[width=8.2cm,height=6.7cm]{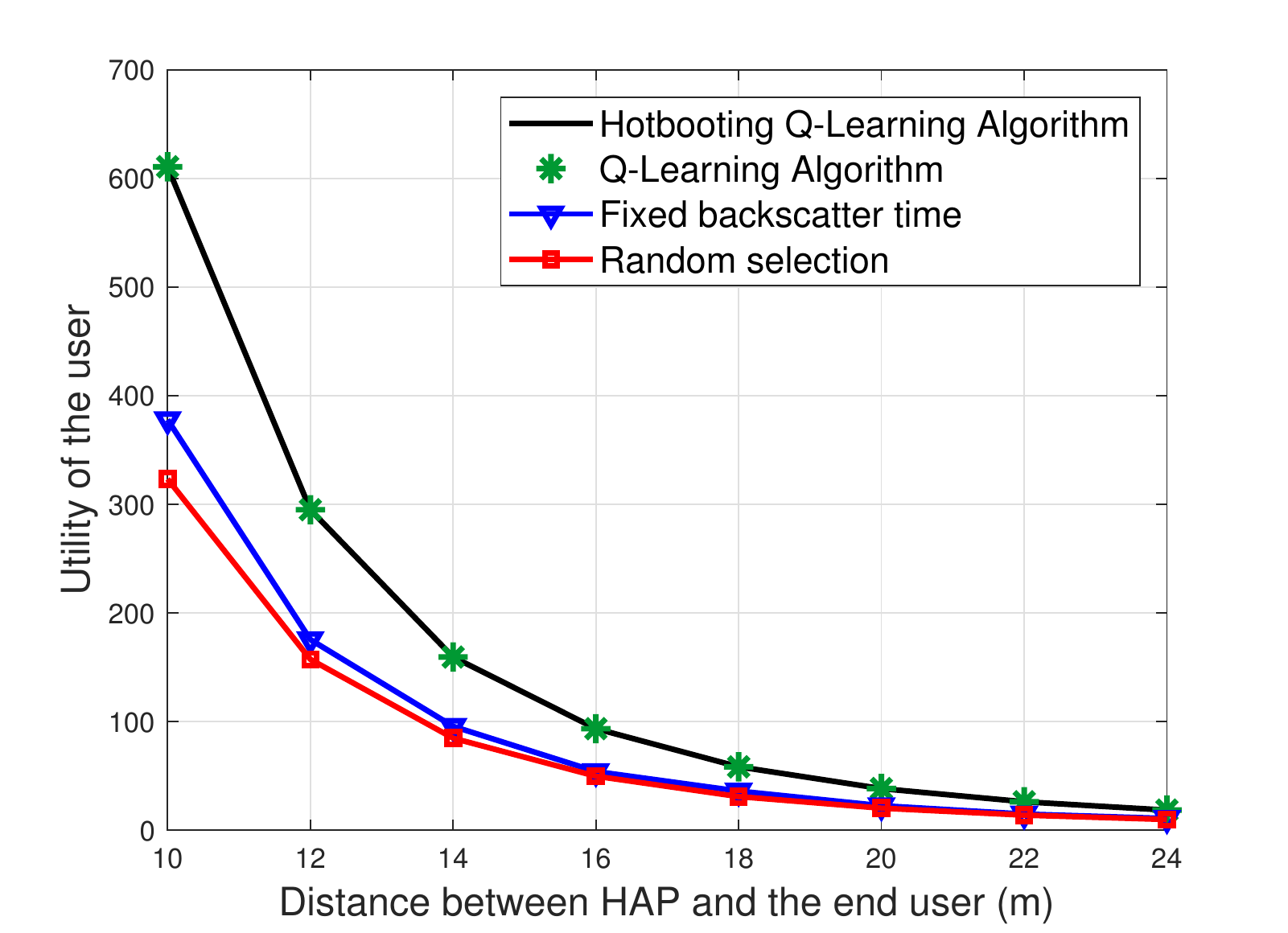}
	\centering
	\caption{Average utility of the user vs. the distance between the user and HAP.}
	\label{dist}
\end{figure}

\vspace{-3mm}
\section{Conclusion and Future work}\label{con}

In this paper, we have formulated an interference avoidance game for an ambient backscatter communication system in which the user determines its portion of time for backscattering and the smart interferer attempts to choose its jamming power to interrupt the transmission. Due to lack of information about the system state and the smart interferer transmission power, the interaction between the user and the interferer is  modeled as a dynamic game and the optimal strategies are  obtained using both analytical results and Q-learning. In order to expedite the convergence, we have adopted a hotbooting Q-learning that can learn faster. Based on the simulation results, it can be observed that the proposed approach can significantly improve the user utility in comparison to the random and fixed schemes. Considering other parameters of the user and the smart interferer such as mobility and their location as variables, the equilibrium of the  game can be further optimized. However, since the number of the variables in the Q-table get increased, deep Q-learning can be employed to approximate the Q-table  and reduce the complexity of the  learning game; this is left as a future work.
\begin{figure}[!t]
	\includegraphics[width=8.2cm,height=6.7cm]{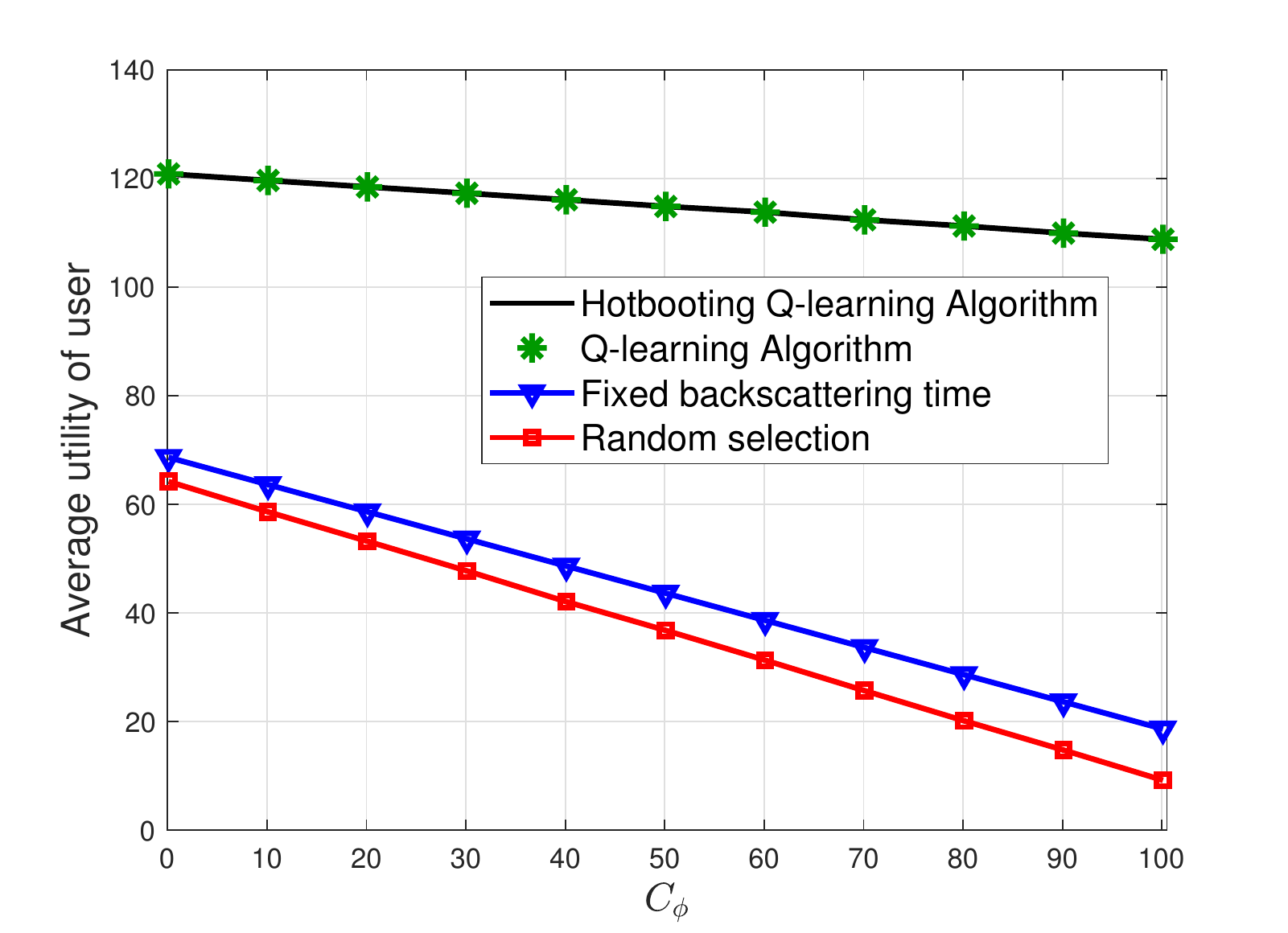}
	\centering
	\caption{Average utility of the user vs. cost for backscattering time.}
	\label{cost}
\end{figure}

	

\bibliographystyle{ieeetr}
\bibliography{refs}

\end{document}